%% file: main.tex
\documentclass[11pt]{article}
\usepackage[colorlinks]{hyperref}
\hypersetup{linkcolor=cyan,filecolor=cyan,citecolor=cyan,urlcolor=cyan}
\usepackage{xspace}
\usepackage{fullpage}
\usepackage{boxedminipage}
\usepackage[boxed]{algorithm}
\usepackage{epigraph}

\usepackage[sc]{mathpazo}
\usepackage{amsmath}
\usepackage{mathtools}
\usepackage{framed}
\usepackage[framemethod=tikz]{mdframed}
\usepackage{titlesec}
\usepackage{lipsum}%

\usepackage{tikz}
\usetikzlibrary{shapes.geometric}
\usetikzlibrary{arrows}
\usetikzlibrary{arrows.meta}
\usetikzlibrary{patterns}
\usetikzlibrary{shapes.misc}

\newcommand{\polylog}{\mathsf{polylog}}

\renewcommand{\paragraph}[1]{\vspace{0.15cm}\noindent {\bf #1}}

\input{headers}

\title{Parallel Balanced Allocations:\\ The Heavily Loaded Case}

\author{
 Christoph Lenzen\thanks{MPI for Informatics, Germany. Email: 
 \texttt{clenzen@mpi-inf.mpg.de}.}
 \and
Merav Parter\thanks{Weizmann Institute, Israel. Email: 
\texttt{merav.parter@weizmann.ac.il}. Supported in part by grants from BSF-NSF 
No.\ 2017758 and  
Minera  No.\ 713238.}
\and
Eylon Yogev\thanks{Technion, Israel. Email: \texttt{eylony@gmail.com}. 
Supported by the European Union's Horizon 2020 research and 
innovation program under grant agreement No.\ 742754.}
}

\date{}

\begin{document}
\maketitle
\input{abstract}

\thispagestyle{empty}
\newpage

\setcounter{page}{1}
\input{intro}
\input{prelim}

\input{parallel-algorithm}

\input{lower}
\input{lower-simulation}
\input{assym.tex}
\input{conclusion}

\bibliographystyle{alpha} 
\bibliography{crypto}

\end{document}

%% file: headers.tex
\usepackage{amsthm,amsmath,amssymb}
\usepackage{cleveref,aliascnt}
\newtheorem{proposition}{Proposition}
\newtheorem{theorem}{Theorem}

\newtheorem{definition}{Definition}

\newtheorem{lemma}{Lemma}

\newtheorem{claim}{Claim}

\newtheorem{corollary}{Corollary}

\usepackage{tikz}
\usepackage{relsize}
\usepackage{ctable}

\crefname{claim}{Claim}{Claims}
\crefname{observation}{Observation}{Observations}

\newcommand{\A}{{\mathcal A}}

\newcommand{\E}[1]{\mathbf E\! \left[ {#1} \right]}

\newcommand{\N}{{\mathbb{N}}}
\newcommand{\R}{{\mathbb{R}}}

\newcommand{\eg}  {e.g.,\ }

\newcommand{\etal}{{et~al.\ }}

\newcommand{\ignore}[1]{}

\newcommand{\m}{\widetilde{m}}

%% file: abstract.tex
\begin{abstract}
We study parallel algorithms for the classical balls-into-bins problem, in which $m$ balls acting in parallel as separate agents are placed into $n$ bins. Algorithms operate in synchronous rounds, in each of which balls and bins exchange messages once. The goal is to minimize the maximal load over all bins using a small number of rounds and few messages. 

While the case of $m=n$ balls has been extensively studied, little is known about the heavily loaded case. In this work, we consider parallel algorithms for this somewhat neglected regime of $m\gg n$. The na\"ive solution of allocating each ball to a bin chosen uniformly and independently at random results in maximal load $m/n+\Theta(\sqrt{m/n\cdot \log n})$ (for $m\geq n \log n$) with high probability (w.h.p.). In contrast, for the sequential setting Berenbrink et al.~\cite{BerenbrinkCSV06} showed that letting each ball join the least loaded bin of two randomly selected bins reduces the maximal load to $m/n+O(\log\log m)$ w.h.p. To date, no parallel variant of such a result is known.

We present a simple parallel \emph{threshold} algorithm that obtains a 
maximal load of $m/n+O(1)$ w.h.p.\ within $O(\log\log 
(m/n)+\log^* n)$ rounds. The algorithm is \emph{symmetric} (balls and bins all 
``look the same''), and balls send $O(1)$ messages in expectation. 
The additive term 
of $O(\log^* n)$ in the complexity is known to be tight for such 
algorithms~\cite{LenzenW16}. We also prove that our analysis is tight, i.e., 
algorithms of the type we provide must run for $\Omega(\min\{\log\log 
(m/n),n\})$ rounds w.h.p.

Finally, we give a simple \emph{asymmetric} algorithm (i.e., balls are aware of 
a common labeling of the bins) that achieves a maximal load of $m/n + O(1)$
in a {\em constant} number of rounds w.h.p. Again, balls send only a single 
message per round, and bins receive $(1+o(1))m/n+O(\log n)$ messages w.h.p.
This goes to show that, similar to the case of $m=n$, asymmetry allows for highly efficient solutions.
\end{abstract}

%% file: intro.tex
\section{Introduction}
We consider simple parallel algorithms for the heavily loaded regime of the well-known balls into bins problem. When $m$ balls are thrown randomly into $n$ bins, the maximal load can be bounded by $m/n+\Theta(\sqrt{\log n \cdot m/n})$ \emph{with high probability (w.h.p.)}\footnote{Throughout this work, we say that an event happens with high probability if it succeeds with probability of at least $1-1/n^c$ for any constant $c\geq 1$.} for any $m=\Omega(n\log n)$ (e.g., by Chernoff's bound). In the balanced case, i.e., for $m=n$, it was demonstrated that parallel communication between balls and bins can considerably improve this load using a small number of messages and rounds. In contrast, for the $m \gg n$ regime, to this point, there was no (communication efficient) parallel algorithm that outperforms the na\"ive random allocation. 

In this paper, we ask how to leverage communication to improve the maximal load 
for this heavily loaded case. We are in particular intrigued by the number of 
communication rounds required to achieve the almost perfect maximal load of 
$m/n+O(1)$. We focus primarily on algorithms which are symmetric (bins are 
anonymous) and use few messages. 

\paragraph{The Classical Setting of Balls into Bins.} 
Balls into bins and related problems have been studied thoroughly in 
a wide range of models. The high-level goal of any balls-into-bins 
algorithm is to allocate ``efficiently'' a set of items (e.g., jobs, balls) 
to a set of resources (machines, bins). The na\"ive single-choice algorithm 
places each ball into a bin chosen independently and uniformly at random.  It 
is well-known that for $m=n$ this achieves a maximal load of $O(\log n/\log\log n)$ with 
high probability.
In a seminal work, Azar \etal introduced the multiple-choice paradigm, in which the 
balls are placed into bins \emph{sequentially} one by one, and each ball is 
allocated to the least loaded among $d\geq 2$ randomly selected bins. They 
showed that this algorithm achieves, w.h.p., a maximal load of $O(1+\log\log 
n/\log d)$, an exponential improvement over the single choice algorithm.

Adler et al.~\cite{AdlerCMR98}
introduced the parallel framework for the balls-into-bins problem, with the objective of parallelizing this sequential 
multiple choice process. They restricted attention to simple and natural parallel algorithms that are both (i) symmetric: all balls and bins run the same algorithms, and bins are anonymous; and (ii) non-adaptive: each ball picks a set of $d$ bins uniformly and independently at random and communicate only with these bins throughout the protocol. They showed that such symmetric and non-adaptive algorithms can achieve a total load of $\Theta(\log\log n/\log\log \log n)$ with the same number of rounds. 

Lenzen and Wattenhofer \cite{LenzenW16} relaxed the non-adaptivity constraint, and presented an adaptive and symmetric algorithm that obtains a bin load of $2$, w.h.p., within $O(\log^* n)$ rounds and using a total of $O(n)$ messages. Again, this is tight for this class of algorithms, and dropping any of the constraints the lower bound imposes leads to constant-round solutions.

\paragraph{The Heavily Loaded Case of Balls into Bins.} 
It has been noted in the literature that the
$m\gg n$ regime of the balls into bins problem is fundamentally different than 
the case where $m=n$; this explains why attempts to extend the analysis of 
existing $m=n$ algorithms to the heavily loaded case mostly fail 
\cite{BerenbrinkCSV06,TalwarW14}. In a breakthrough result, Berenbrink et 
al.~\cite{BerenbrinkCSV06} provided an ingenious analysis for the multiple 
choice process in the heavily loaded regime. They showed that when balls are 
allowed to pick the best among $2$ random choices, the bin load becomes 
$m/n+O(\log\log n)$ with high probability. Thus the $2$-choice process 
super-exponentially improves the excess bin load compared to the single choice random 
allocation \emph{and} makes it independent of $m$.

To the best of our knowledge, there has been no work that parallelizes this 
sequential process in a similar manner as has been done by Adler et al.\ 
and others for the $m=n$ case.\footnote{We note that Stemann~\cite{stemann96parallel} considers the possibility that $m>n$, but provides algorithms for load $O(m/n)$ only; for almost the entire range of parameters, the na\"ive algorithm or using multiple instances of algorithms for $m\leq n$ yields better results.}
As a result, no better parallel algorithm has been known for this regime other than placing balls randomly into bins.

\paragraph{Our Results.}
We propose a very simple \emph{threshold algorithm} (cf.~\cite{AdlerCMR98}) that appears to be suitable for the heavily loaded regime. In every synchronous round $r$ of our algorithm, each unallocated ball sends a join-request to a bin chosen uniformly at random. Bins will accept balls up to a load of $T_r$ (a threshold that increases with $r$). Thus, a bin with load $\ell$ at the beginning of round $r$ acknowledges up to $T_r-\ell$ requests (chosen arbitrarily among all received requests) and declines the rest. We show that such a simple algorithm achieves a maximal load of $m/n+O(1)$ within $O(\log\log (m/n))$ rounds with high probability.
\begin{theorem}\label{thm:upper}
There exists a parallel symmetric and adaptive algorithm of $O(\log\log (m/n) + 
\log^*n)$ rounds that achieves maximal load of $m/n+O(1)$ with high probability. The algorithm uses a total of $O(m)$ messages, w.h.p.
\end{theorem}
Note that, trivially, one can place all balls within $n$ rounds, by each ball approaching each bin once (and bins using thresholds of $L_r=\lceil m/n \rceil$ in all rounds). Thus the above time bound is of interest whenever $\log \log (m/n) \gg n$.

The technically most challenging part is our lower bound argument. 
We consider a special class of threshold algorithms to which our algorithm belongs. This class consists of all threshold algorithms in which in every round, every (unallocated) ball contacts $O(1)$ bins sampled uniformly and independently at random. This class generalizes our algorithm in two ways. First, it allows a ball to contact $O(1)$ bins per round instead of only $1$ (as in the main phase of our algorithm). Second, it allows bins to have distinct threshold values, which can depend on the state of the entire system in an arbitrary way. 
\begin{theorem}\label{thm:lower}
Any threshold algorithm in which in each round balls choose $O(1)$ bins to 
contact uniformly and independently at random w.h.p.\ runs for $\Omega(\min\{\log\log 
(m/n),2^{n^{\Omega(1)}}\})$ rounds or has a maximal load of $m/n+\omega(1)$.
\end{theorem}
This theorem applies to the algorithm of \Cref{thm:upper}, but not to the trivial $n$-round algorithm mentioned above. We conjecture that any threshold algorithm runs for $\Omega(\min\{\log\log (m/n),n\})$ rounds or incurs larger loads, but a proof seems challenging due to the obstacles imposed by balls using differing probability distributions for deciding which bins to contact.

\paragraph{Asymmetric Algorithms.}
In the asymmetric setting, all bins are distinguished based on globally known IDs, which can be rephrased as all balls' port numberings of bins being consistent. A perfect allocation can be obtained trivially in this setting, simply by letting all balls contact the first bin, which then can send to each ball the bin ID to which it should be assigned. To rule out such trivial solutions, one should restrict attention to algorithms in which no bin receives (significantly) more messages than necessary. Concretely, bins should receive no more than $(1+o(1))m/n+O(\log n)$ messages; as with constant probability some bin will receive $m/n+\sqrt{m/n}+\log n$ messages even if each ball sends a single message, this is the best we can hope for.
\begin{theorem}\label{thm:assym}
There exists a parallel asymmetric algorithm that achieves a maximal load of 
$m/n+O(1)$ within $O(1)$ rounds w.h.p., where each bin receives a total of
$(1+o(1))m/n+O(\log n)$ messages w.h.p. 
\end{theorem}
This goes to show that, similar to the case of $m=n$, asymmetry allows for 
highly efficient solutions. In what follows, we give a high-level overview of 
the proofs of \Cref{thm:upper,thm:lower}. The full proof of \Cref{thm:assym} is 
given in \Cref{sec:apx-asym}.

\paragraph{Additional Related Work.}
Following \cite{AzarBKU99}, multiple-choice algorithms 
have been studied extensively in the sequential setting. For instance, 
\cite{Vocking03} considered a variant of this setting where the selections made 
by balls are allowed to be nonuniform and dependent. The works 
\cite{shah2002use,MitzenmacherPS02} have studied the effect of memory when 
combined with the 
multiple choice paradigm and showed that a choice from memory is 
asymptotically better than a random choice. The analysis of the multiple choice 
process for the heavily loaded case was first provided by 
\cite{BerenbrinkCSV06} and considerably simplified by \cite{TalwarW14}. See 
\cite{Wieder17} for a survey on sequential multiple-choice algorithms.

Turning to the distributed/parallel setting, \cite{sauerwald2012tight} studied distributed load balancing protocols on general graph topologies. \cite{berenbrink2012multiple} considers a semi-parallel framework for balls into bins, in which the balls arrive in batches rather than one by one as in the sequential setting. 
\cite{AlistarhDRS14} consider a variant of the balls-into-bins problem, namely, the renaming problem and the setting of synchronous message passing with failure-prone servers. Finally, \cite{bertrand20141} introduced a general framework for parallel balls-into-bins algorithms and generalizes some of the algorithms analyzed in \cite{LenzenW16}.

\subsection{Our Approach in a Nutshell.}
\paragraph{The Symmetric Algorithm.}
To get some intuition on threshold algorithms, we start by considering the most 
na\"ive algorithm, in which each bin agrees to accept at most $T=m/n+O(1)$ 
balls in total, without modifying its threshold over the course of the 
algorithm. That is, in every round each unallocated ball picks a bin uniformly 
and independently at random, each bin agrees to accept at most $T$ balls in 
total, and rejects the rest. Clearly, the final load of each bin is bounded by 
$T$ and hence it remains to consider the running time of such an algorithm. One 
can show that, w.h.p., after a single round a constant fraction of the bins are 
going to be full (i.e., contain $T$ balls). Hence, the probability of an 
unallocated ball to contact a full bin in the following rounds is constant. 
This immediately entails a running time lower bound of $\Omega(\log n)$, even 
if the balls may contact a constant number of bins per round.

The crux idea of our symmetric algorithm is to set the threshold \emph{lower} 
than the
allowed bin load (\eg in the first round we set $T=m/n-(m/n)^{2/3}$).  At 
first glance, this seems 
unintuitive as a
bin might reject balls despite the fact that it still has room. The key
observation here is that setting the threshold a bit smaller than the allowed
load keeps all bins equally loaded throughout the algorithm, yet permits placing
all but a few of the remaining balls in each step. This prevents the situation
where an unallocated ball blindly searches for a free bin in between many 
occupied
bins. Crunching the numbers shows that this approach reduces the number of 
remaining
balls to $O(n)$ in $O(\log \log (m/n))$ rounds, after which the established 
techniques
for the case of $m=n$ can be applied.

\paragraph{The Lower Bound.}
Our lower bound approach considers a natural family of threshold algorithms, 
which in particular captures the above algorithm. Every algorithm is this 
family has the following structure. In each round $i$, every unallocated ball 
picks $O(1)$ bins independently and uniformly at random. Every bin $j$ accepts 
up to $T_{i,j}$ requests and rejects the rest. The value $T_{i,j}$ can be 
chosen non-deterministically by the bins.

This class is more general than our algorithm, in several ways. Most significantly, it allows 
bins to have different thresholds. The decision of these can depend on the 
system state at the beginning of each round (excluding future random choices of balls). 
Moreover, we allow for algorithms that ``collect'' allocation requests from balls for several 
rounds before allocating them according to the chosen threshold. While this is not a 
good strategy for algorithms, is it useful in the simulation part of the proof, which is 
explained next.

The proof follows in two steps. First, we prove the lower bound for degree one 
algorithms (where balls contact a single bin in each iteration) in the family 
described above. 
The argument for this step is somewhat technical, and 
it is based 
on focusing on one class of bins that have roughly the same number of rejected 
balls in expectation. We show that one can find such a class of bins which 
captures a large fraction of the expected number of rejected balls. We then 
exploit the fact that all bins in this class are roughly the same, which allows 
us to provide concentration results for that class. 

The second step is a simulation technique in which we show how to 
simulate an algorithm with higher degree by an algorithm from the above family. 
Roughly speaking, we simulate a degree $d$ algorithm by contacting a single bin 
over $d$ different rounds. Only after these $d$ rounds the bins decide which balls to 
accept. Here we crucially rely on the fact that our lower bound for single degree 
algorithms includes such algorithms.

%
%

%% file: prelim.tex
\section{Preliminaries}
\begin{definition}[With high probability (w.h.p.)]
We say that the random variable $X$ attains values from the set $S$ with high 
probability, if $\Pr[X \in S] \ge 1 - 1/n^c$ for an arbitrary, but
fixed constant $c > 0$. More simply, we say $S$ occurs w.h.p.
\end{definition}
We use some theory on negatively associated random variables, which is 
given in \cite{DubhashiR98}.

\begin{definition}[Negative Association]\label{def:na}
	A set of random variables $X_1,\ldots,X_n$ is said to be negatively 
	associated 
	if for any two disjoint index sets $I,J \subseteq [n]$ and two functions 
	$f, g$ 
	that are both monotone increasing or both monotone decreasing, it holds that
	\begin{align*}
		\E{f(X_i: i \in I) \cdot g(X_j : j \in J)} \le 
		\E{f(X_i : i \in I)} \cdot 	\E{g(X_j : j \in J)}.
	\end{align*}
\end{definition}

\begin{lemma}[Chernoff Bound]\label{lem:chernoff}
Let $X_1,\ldots,X_m$ be independent or negatively associated 
random variables that take the value 
1 with probability $p_i$ and 0 otherwise, $X=\sum_{i=1}^{m}X_i$, and 
$\mu=\E{X}$. Then for any $0 < \delta < 1$,
$$\Pr[X < (1-\delta)\mu] \le e^{-\delta^2 \mu /2}$$
and
$$\Pr[X > (1+\delta)\mu] \le e^{-\delta^2 \mu /3}.$$
If $\mu > 2\log m$, with $\delta = \sqrt{2\log m / \mu}$ we get
that
$$\Pr[X < \mu - \sqrt{2\mu\log m}] \le 1/m,$$
and
$$\Pr[X > \mu + \sqrt{3\mu\log m}] \le 1/m.$$
\end{lemma}

\begin{proposition}[\cite{DubhashiR98}, Proposition 
7(2)]\label{prop:monotone-na}
Non-decreasing (or non-increasing) functions of disjoint subsets of
negatively associated variables are also negatively associated.
\end{proposition}

Our lower bound proof makes use of the following Berry-Esseen 
inequality. 
\begin{theorem}[Berry-Esseen 
Inequality~\cite{berry1941accuracy,Esseen1942}]\label{theorem:berry-esseen}
Let $Y_j$, $j\in \{1,\ldots,M\}$, be i.i.d.\ random variables with
$\E{Y_j}=0$, $\sum^2\coloneqq \E{|Y_j|^2}>0$, and $\rho\coloneqq
\E{|Y_j|^3}<\infty$, and let $Y=\sum_{j=1}^M Y_j$. Denote by $F$ the 
cumulative distribution functions of $\frac{Y}{\sum \sqrt{M}}$
and by $\phi$ the cumulative distribution function of the standard normal 
distribution. Then
\begin{equation*}
\sup_{s\in \R}\{|F(s)-\phi(s)|\}\leq \frac{c \rho}{\sum^3 \sqrt{M}}\ , 
\mbox{~~for a constant $c$}.
\end{equation*}
\end{theorem}
\paragraph{Symmetric Algorithm for $m=n$.}
Our algorithm for the heavily loaded regime uses the algorithm of \cite{LenzenW16} for allocating $n$ balls into $n$ bins. We denote this algorithm by $\A_{light}$. Specifically, we use the following theorem.
\begin{theorem}\label{thm:lenzen}[From {\cite{LenzenW16}}]
There exists a symmetric algorithm for placing $n$ balls into $n$ bins with the 
following properties w.h.p.: 
The algorithm terminates after $\log^* n + O(1)$ rounds with bin load at most 
$2$. 
The total number of messages sent is $O(n)$, where in each round 
balls send and receive $O(1)$ messages in expectation and 
	$O(\log n)$ many with high probability. Finally, in each round, 
bins send and receive $O(1)$ messages in expectation and $O(\log n/ 
	\log \log n)$ many with high probability.
\end{theorem}

%% file: parallel-algorithm.tex
\section{The Parallel Symmetric Algorithm}\label{sec:parallel}
In this section, we describe our symmetric algorithm for allocating $m$ balls 
into $n$ bins. We begin by describing the precise model in which the algorithm 
works.

\paragraph{The Model.}
The system consists of $m$ bins and $n$ balls, and operates in the synchronous 
message passing model, where each round consists of the following steps.
\begin{enumerate}
\item Balls perform local computations and send messages to arbitrary bins.
\item Bins receive these messages, perform local computations and send messages 
to any balls they have been contacted by in this or earlier rounds.
\item Balls receive these messages and may commit to a bin (and terminate).
\end{enumerate}
All algorithms may be randomized and have unbounded computational resources;
however, we strive for using only very simple computations.

\paragraph{High-Level Description.}
The algorithm consists of two phases. The first phase consists of $O(\log\log (m/n))$
rounds, at the end of which the number of unallocated balls is $O(n)$. The
second phase consists of $O(\log^* n)$ rounds and completes the allocation by
applying \Cref{thm:lenzen}~\cite{LenzenW16}.

For simplicity, we will assume that all values specified in the following are integers; as we aim for asymptotic bounds, rounding has no relevant impact on our results.
In our algorithm, the threshold values of all bins are the same, but depend on the current round.
In the first round, all bins set their threshold to $T=m/n-(m/n)^{2/3}$, each ball picks a single bin uniformly at random, and bins accept at most $T$ balls and reject the rest. 
Applying Chernoff's bound, we see that w.h.p.\ each bin is contacted by at 
least $m/n-\sqrt{10\log n\cdot m/n}>T$ balls. Hence, each bin has exactly $T$ 
allocated balls after the first round. Accordingly, the number of unallocated 
balls after the first round is $m'=m-T\cdot n=O(m^{2/3}n^{1/3})$. We continue 
the same way in the second round, handling an instance with $m'$ balls and $n$ 
bins. It follows that the number of remaining balls after $i$ rounds is bounded 
by $O(m^{(2/3)^i} n^{1-(2/3)^i})$. When $m'$ gets very close to $n$, i.e., 
$m'\in n\polylog(n)$, concentration is not sufficiently strong any more to 
guarantee that \emph{all} bins receive the desired number of balls. However, 
one can show that w.h.p.\ this holds true for the vast majority of bins. 
Overall, we show that after $O(\log\log (m/n))$ rounds, $O(n)$ unallocated 
balls remain. 

At this point, we employ the parallel algorithm of Lenzen and Wattenhofer \cite{LenzenW16}, which takes additional $O(\log^* n)$ rounds. To this end, we let each bin act as $O(1)$ virtual bins. This way, at most $O(1)$ additional balls will be allocated to each bin, as the algorithm guarantees a maximum bin load of $2$. We next describe the algorithm and its analysis in detail.

\paragraph{The Algorithm $\A_{heavy}$:}
\begin{enumerate}
\item Set $\m_0=m$.
\item For $i=0,\ldots, O(\log \log (m/n))$ do:
\begin{enumerate}
	\item Each ball sends an allocation request to a uniformly sampled bin.
	\item Set $T_i=\frac{m}{n} - (\frac{\m_i}{n})^{\frac{2}{3}}$. Each bin accepts up to $T_i-\ell_i$ balls, where $\ell_i$ is the load of the bin at the beginning of the round.
	\item Set $\m_{i+1} = \m_i^{2/3}n^{1/3}$.
\end{enumerate}
	\item At this point at most $O(n)$ balls are unallocated (w.h.p.). Run 
	$\A_{light}$ for the remaining balls with each bin simulating $O(1)$ virtual 
	bins.
\end{enumerate}

\begin{theorem}\label{thm:upper2}
Algorithm $\A_{heavy}$ finishes after $O(\log \log (m/n)+ \log^* n)$ rounds 
with maximal load of $m/n + O(1)$, w.h.p., using in total $O(m)$ messages (over all rounds).
Each ball sends and receives $O(1)$ messages in expectation and $O(\log n)$ many 
w.h.p. Each bin sends and receives $(1+o(1))m/n+O(\log n)$ messages w.h.p.
\end{theorem}

\begin{proof}
For any round $i$ of step (2), let $m_i$ be the number of unallocated balls at 
the beginning of the round, and notice that $\m_i$ is the bin's estimate of 
$m_i$. Fix a round $i$. Let $X_b$ be a random variable indicating the number of balls 
that choose bin $b$ in round $i$ (we suppress the round index for ease of notation) and set $T_{-1}\coloneqq 0$.

Observe that $(\m_i/n)^{2/3}=\m_{i+1}/n$. Moreover, $m_i\geq \m_i$, as $n T_{i-1}=m-n(\m_{i-1}/n)^{2/3}=m-\m_i$ balls can be allocated by the end of round $i-1$. We make frequent use of these observations in the following. We start by bounding the probability that a bin gets ``underloaded'' in a given round, i.e., despite the conservatively small chosen threshold, it does not receive sufficiently many requests to allocate $T_i-T_{i-1}$ balls in round $i$.
\begin{claim}\label{cl:full}
$P[X_b< T_i-T_{i-1}]<e^{-(\frac{\m_i}{n})^{1/3}/2}$.
\end{claim}
\begin{proof}
For all $i$, it holds that
\begin{equation*}
T_i-T_{i-1} = \frac{\m_i}{n}-\frac{\m_{i+1}}{n}=\frac{\m_i}{n}-\left(\frac{\m_i}{n}\right)^{\frac{2}{3}}\,.
\end{equation*}
As $m_i\geq \m_i$, $\E{X_b} = \frac{m_i}{n}\geq \frac{\m_i}{n}$. Using a Chernoff bound with $\delta = (\frac{m_i}{n})^{-1/3}$, we get that
\begin{align*}
\Pr\left[X_b < T_i\right] &\le \Pr\left[X_b < \frac{m_i}{n} - 
\left(\frac{m_i}{n}\right)^{\frac{2}{3}}\right] \\
&= \Pr\left[X_b < (1-\delta)\E{X_b}\right] \le e^{-\delta^2\E{X_b}/2} \\
& = e^{-(\frac{\m_i}{n})^{1/3}/2}\,.\qedhere
\end{align*}
\end{proof}
Using this bound, we next show that each bin is allocated balls to match its 
threshold in each round, at least until only $n \polylog(n)$ balls remain.
\begin{claim}\label{cl:all}
Let $i_0\in O(\log \log (m/n))$ be minimal with the property that $\m_{i_0}\leq n c^3\log^3 n$ for a sufficiently large constant $c$. Then $m_{i_0}=\m_{i_0}$ w.h.p.
\end{claim}
\begin{proof}
We apply \Cref{cl:full} to all bins and all $i<i_0$. Using a union bound over all such events, the probability that $X_b<T_i-T_{i-1}$ in any such round for any bin is bounded by
\begin{align*}
\sum_{i=0}^{i_0-1}n e^{-(\frac{\m_i}{n})^{1/3}/2}\in 
O\left(n\sum_{i=0}^{i_0-1}2^{-i}e^{-\left(\frac{\m_{i_0-1}}{n}\right)^{1/3}/2}\right)
\subseteq n e^{-\Omega(c\log n)} \subseteq n^{-\Omega(c)}\,.
\end{align*}
Thus, w.h.p.\ each bin has exactly $\sum_{i=0}^{i_0-1}T_i = m/n-\m_{i_0}/n$ 
balls allocated to it at the end of round $i_0-1$. Therefore, 
$m_{i_0}=\m_{i_0}$ w.h.p. 
\end{proof}

It remains to consider the final $O(\log \log \log n)$ iterations required to reduce $\m_i$ to $O(n)$. As the number of balls is not large enough anymore to ensure sufficient concentration for individual bins, we consider the random variable $Y_i$ counting the number of balls allocated to all bins together in round $i$.
\begin{claim}\label{cl:most}
Let $i_1$ be minimal with the property that $\m_{i_1}\leq 2n$. For each round $i_0\leq i<i_1$ and any $c>0$, it holds that $Y_i\geq n\left(T_i-T_{i-1}-f(c)2^{-(i_1-i)}\right)$ with probability at least $1-n^{-c}$, where $f\colon \mathbb{R}^+\to \mathbb{R}^+$.
\end{claim}
\begin{proof}
Denote by $Z_b$, $b\in \{1,\ldots,n\}$, the indicator variables which are $1$ 
if bin $b$ receives fewer than $T_i-T_{i-1}$ allocation requests in round $i$ 
and $0$ else. By \Cref{cl:full} and linearity of expectation, we have for 
$Z=\sum_{b=1}^n Z_i$ that
\begin{equation*}
\E{Z}\leq e^{-(\frac{\m_i}{n})^{1/3}/2} n\,.
\end{equation*}
The random variables $Z_b$ are negatively associated (according to \Cref{def:na}). To see this, observe that by \cite[Theorem 13]{DubhashiR98} we know that $X_1,\ldots,X_n$ are 
negatively associated: the $Z_b$ are monotone nonincreasing functions of disjoint subsets of the negatively associated variables $X_1,\ldots, X_n$ (namely, $Z_b$ is a function of the set $\{X_b\}$), so \Cref{prop:monotone-na} applies. Therefore, we can apply a Chernoff bound (with $\delta=1$) to $Z$:
\begin{align*}
\Pr\left[Z > 2\E{Z}\right]  \le e^{-\E{Z}/3}\,.
\end{align*}
If $\E{Z}\geq 3c\log n$ for a sufficiently large constant $c$, this entails 
that $Z\leq 2\E{Z}$ w.h.p. Otherwise, we use a simple domination argument: each 
$Z_b$ is replaced by an independent 0-1 variable $Z_b'$ that is $1$ with 
probability $3c\log n/n$, so that for $Z'\coloneqq \sum_{b=1}^n Z_b'$ we have 
that
\begin{align*}
\Pr\left[Z > 2\E{Z'}\right]  \le \Pr\left[Z'>2\E{Z'}\right] \le e^{-c} < n^{-c}\,.
\end{align*}
Together, this entails that $Z\leq 6c\log n + 2e^{-(\frac{\m_i}{n})^{1/3}/2} n$ 
w.h.p. As $i\geq i_0$ (where $\m_{i_0}\in n\polylog(n)$), we have that 
$2^{i_1-i}\in 2^{O(\log \log \log n)}$ and $T_i-T_{i-1}\leq T_i \in 
\polylog(n)$. Hence $6c \log n (T_i-T_{i-1})<f(c) 2^{-(i_1-i+1)}n$ for a 
suitable choice of $f$. As $2e^{-(\frac{\m_i}{n})^{1/3}/2}$ decreases 
exponentially in $\m_i/n$, which itself decreases exponentially in $i$, we also 
have that
\begin{align*}
2e^{-(\frac{\m_i}{n})^{1/3}/2}(T_i-T_{i-1}) n < 2e^{-(\frac{\m_i}{n})^{1/3}/2} 
\frac{\m_i}{n}  <f(c) 2^{-(i_1-i+1)}n
\end{align*}
if $f(c)$ is sufficiently large. Noting that $Y_i\geq (T_i-T_{i-1})(n-Z)$, the claim follows.
\end{proof}

\begin{claim}
For any $c>0$, $m_{i_1}\leq g(c)n$ with probability at least $1-n^{-c}$, where $g\colon \mathbb{R}^+\to \mathbb{R}^+$. 
\end{claim}
\begin{proof}
The number of unallocated balls at the beginning of round $i_1$ is $m_{i_1} = m 
- \sum_{i=0}^{i_1-1} Y_i$. By \Cref{cl:all}, we have that $m_{i_0}=\m_{i_0}$ 
w.h.p., i.e., $Y_i=(T_i-T_{i-1})n$ for all $i<i_0$ w.h.p. For $i_0\leq i < 
i_1$, by \Cref{cl:most} we have that $Y_i\geq 
n\left(T_i-T_{i-1}-f(c)2^{-(i_1-i)}\right)$ w.h.p., where $c$ is the constant 
in the w.h.p.\ bound. Accordingly, by a union bound it holds that 
\begin{align*}
m_{i_1}  \leq m - n\left(\sum_{i=0}^{i_1-1}(T_i-T_{i-1}) + 
\sum_{i=i_0}^{i_1-1}f(c)2^{-(i_1-i)}\right)  < m - (\m_0 - \m_{i_1}) + f(c)n 
\leq (2+f(c))n
\end{align*}
with probability $1-(i_1-i_0+1)n^{-c}$. As $i_1-i_0\in O(\log \log \log n)$, $m_{i_1}\leq g(c)n$ w.h.p.\ for a suitable choice of $g$.
\end{proof}

Thus, after $i_1\in O(\log \log (m/n))$ iterations, at most $g(c)n$ balls remain unallocated w.h.p. We apply 
$\A_{light}$, where each of the $n$ bins simulates $g(c)$ virtual bins. That is, 
any ball allocated in one of the $g(c)$ virtual bins will be allocated in the real 
bin. Finally, by the properties of $\A_{light}$ we have that each virtual bin will 
have at most $2$ balls and thus each real bin 
will add at most $2g(c)$ balls. Overall, the total load of any bin is $m/n + O(1)$.

\paragraph{Number of Messages.}
We bound the number of messages sent by balls and bins. The number of messages 
sent in step 3 is specified in \Cref{thm:lenzen}. Thus, we analyze the messages 
in step 2.

Each ball sends at most 1 message per round, thus a total of $m_i$ in round 
$i$. Each 
round reduces the number of balls by at least a constant factor, cf.~\Cref{cl:all} and~\Cref{cl:most}.
Thus, the total number of messages sent is bounded by a geometric series, i.e., at most $2m$ messages are sent w.h.p.
Moreover, since all balls are identical we have that the 
expected number of message sent by a ball is $O(1)$. The probability that a 
single ball sends more than $\ell$ message is at most $2^{\ell}$. Thus, with 
high probability, a ball sends at most $O(\log n)$ messages. As all messages are sent to uniformly and independently random bins,
a standard Chernoff bound yields that each bin receives $(1+o(1))m/n + O(\log n)$ messages w.h.p.
\end{proof}
\paragraph{A Note on Success Probability.}
As described, Algorithm $\A_{heavy}$ succeeds with high probability in $n$. As $n$ may be a constant, this probability bound could be a constant as well. However, the case of $n < \log \log (m/n)$ can be covered by a trivial algorithm that deterministically guarantees a perfectly balanced allocation in $n$ rounds: balls try all bins one by one, in arbitrary order (which may be different for each ball). Bins use threshold $m/n$ in each round. If $n < \log \log (m/n)$, we can apply this trivial algorithm within our round budget. Combining both algorithms, we achieve a success probability of $1-o(1)$ for the entire parameter range.

%% file: lower.tex
\section{Lower Bound for Threshold Algorithms}
In this section, we present a lower bound for a special class of threshold algorithms. Roughly speaking, the only limitation that we pose here is that in each round unallocated balls pick the bins they contact independently and uniformly at random (as in our upper bound), and bins do not take decisions based on random choices of balls in future rounds. 

This class is more general than our algorithm, as it allows bins to have different thresholds. The decision on these thresholds can be an arbitrary function of the system state at the beginning of the round (excluding future random choices of balls); this does not affect the lower bound result. Moreover, we allow for algorithms that ``collect'' allocation requests from balls for $k\in \mathbb{N}$ rounds before allocating them according to the chosen threshold. While this is not a good strategy for algorithms, it is useful for generalizing our lower bound to algorithms in which balls contact multiple bins in each round, as it allows for a straightforward simulation argument.

\paragraph{The Family of Uniform Threshold Algorithms.}
The degree of an algorithm is the maximal number of bins that a ball contacts in 
a single phase. Formally, in this special threshold model a degree $d$ algorithm collecting for $k$ rounds works in 
phase $i$ as follows. Bins and balls have each an internal state $\sum$. 
Decisions are a function of $\sum$, which is updated after each operation, and 
(private) randomness.
We remark, however, that the structure imposed by the algorithm actually entails 
that the state of a non-allocated ball is simply a function of its own 
randomness only, as it received no information beyond all its requests being 
rejected.

In contrast, bins may perform more complex internal operations. 
Denote by $\ell_b$ the load of bin $b$ at the beginning of phase $i$, i.e., the 
number of balls it has sent accept messages to and which have not yet informed the bin that they are allocated to another bin.
\begin{enumerate}
	\item Each bin $b$ determines its threshold $T_b$ for the current phase. The decision on these thresholds is oblivious to (i.e., stochastically independent from) the random choices of balls in this and future phases.
	\item Based on its state, each ball $u$ chooses (at most) $dk$ bins $b^u_1,\ldots,b^u_{dk}$ uniformly and independently at random to send 
allocation requests to. These requests are sent over $k$ rounds, i.e., at most $d$ per round.
	\item Denote by $R_b$ the set of balls sending a request to bin $b$ in this phase. In the last round of the phase, bin $b$ responds with accept messages to a subset of $R_b$ of size $\max\{T_b-\ell,|R_b|\}$. This set is chosen based on the bin's port numbers for the requesting balls\footnote{For each bin, there is a bijection from $\{1,\ldots,m\}$ to the balls. Requests from a ball are received on the respective port and responses are sent to the same port. Balls have a port numbering of the bins for the same purpose.} and its internal randomness, subject to the constraint that each ball is accepted only once.
	\item Balls receive accept messages. They may decide on an accepting bin to be allocated to (provided they received at least one accept message so far) at the end of \emph{any} phase (i.e., they do not need to commit immediately), where this phase is a function of the phase number in which they received the first accept message.\footnote{This is not a good idea for algorithms, but we use it in our lower bound for a simulation argument.}
	\item Balls that selected a bin inform all bins that sent accept messages to it about its decision at the end of the phase.
\end{enumerate}
For technical reasons, we assume that bins port numbers are chosen adversarially, i.e., first the randomness of balls and bins is determined and then the port numbering is chosen. Algorithms must achieve their load guarantees despite this; note that our algorithms are capable of this.

The structure of this section is as follows. We first establish in \Cref{sec:degone} the lower bound for degree $1$ algorithms, i.e., threshold algorithms in which each unallocated ball contacts \emph{one} bin chosen independently and uniformly and random (our algorithm falls within this class). Then, in \Cref{sec:degd}, we extend the argument to any degree $d$ algorithms for $d=O(1)$ by providing a simulation result. 

\subsection{Lower Bound for Degree 1 Algorithms}\label{sec:degone}
Our lower bound shows that any algorithm in the threshold model, granted that balls 
choose bins uniformly at random, must use a large number of rounds.

\begin{theorem}\label{thm:rejection}
Suppose $M\in \N$ balls each contact one of $2\leq n\in
\N$ bins independently and uniformly at random, where $M\geq Cn$ for a
sufficiently large constant $C$. If bin $i\in \{1,\ldots,n\}$ accepts up to
$L_i$ balls contacting it, where $\sum_{i=1}^n L_i \in M+O(n)$ and $L_i$ does 
not depend on the balls' randomness, with probability at least 
$1-e^{-\Omega((n/t)^{2/3})}$
the number of balls that is not accepted is $\Omega(\sqrt{Mn}/t)$ for $t=\Theta(\min\{\log n, \log(M/n)\})$.
\end{theorem}
\begin{proof}
Denote by $\mu=M/n$ the expected number of messages received by bin $i$. Fix a
bin and denote by $X^{(i)}$ the random variable counting the number of messages it
receives. Because each ball picks a bin uniformly and independently at random,
we have that
\begin{equation*}
X^{(i)} = \sum_{j=1}^M X_j\,,
\end{equation*}
where the $X_j$ are independent $0$-$1$ variables attaining $1$ with probability
$p=1/n\leq 1/2$ (we omit $i$ for ease of notation). Our first goal is to provide a lower bound on the \emph{expected} number of rejected balls. To do that, we first analyze a single bin and show the following:
\begin{claim}\label{cl:onebin}
Any bin has load at least $\mu+2\sqrt{\mu}$ with probability $p_0=\Omega(1)$. 
\end{claim}
\begin{proof}
We apply the Berry-Esseen Inequality (see
\Cref{theorem:berry-esseen}) to the random variables $Y_j\coloneqq X_j-p$, $j\in
\{1,\ldots,M\}$.
Thus, $\sum = \sqrt{p(1-p)}$ and $\rho=p(1-p)(1-2p(1-p))$, yielding that
\begin{align*}
\sup_{x\in \R}\{|F(x)-\phi(x)|\}\leq \frac{c (1-2p(1-p))}{\sqrt{p(1-p)M}} 
\stackrel{p\leq 1/2}{\leq}
\frac{c(1-p)}{\sqrt{p(1-p)M}}<\frac{c}{\sqrt{pM}}\leq\frac{c}{\sqrt{C}}
\end{align*}
in the terminology of the theorem, where $Y=\sum_{j=1}^M X_j-\mu$, i.e., $Y$
equals the deviation of the load of bin $i$ from its expectation. Thus, the
theorem implies that for all $x\geq 0$, we have that
\begin{align*}
\Pr\left[Y\geq x\sqrt{\frac{\mu}{2}}\right]
\stackrel{p\leq 1/2}{\geq} \Pr\left[Y\geq x\sqrt{(1-p)\mu}\right] =
\Pr\left[Y\geq x\sum\sqrt{M}\right] \geq 1-F(x)-\frac{c}{\sqrt{C}}\,.
\end{align*}
Choosing $x=2\cdot \sqrt{2}$ and using that $C$ is sufficiently large, it
follows that
\begin{equation*}
P\left[X^{(i)}\geq \mu+2\sqrt{\mu}\right] \in \Omega(1)\,.\qedhere
\end{equation*}
\end{proof}
Thus, we have shown that any bin has load at least $\mu+2\sqrt{\mu}$ with
probability $p_0\in \Omega(1)$, causing it to reject at least
$\mu+2\sqrt{\mu}-L_i$ balls (provided that $\mu+2\sqrt{\mu}\geq L_i$). 
\begin{corollary}\label{cor:exp}
At least  $p_0 \cdot \sqrt{Mn}$ balls are rejected in expectation for $p_0\in \Omega(1)$.
\end{corollary}
\begin{proof}
By \Cref{cl:onebin}, the expected number of rejected balls for bin $i$ is at least $p_0 \cdot \max\{\mu+2\sqrt{\mu}-L_i,0\}$. Thus, by linearity of expectation the expected number of rejected balls is at least
\begin{align*}
p_0\sum_{i=1}^n \max\{\mu+2\sqrt{\mu}-L_i,0\} \geq p_0\left( M+2\sqrt{M
n}-\sum_{i=1}^n L_i\right) \geq p_0\sqrt{Mn}\,,
\end{align*}
where the final step exploits that $\sqrt{Mn}\geq \sqrt{C}n$ with $C$ being
sufficiently large. 
\end{proof}
So far, we have shown that the expected number of rejected balls is sufficiently large. One of the major obstacles for providing a concentration result comes from the fact that the number of rejected balls might vary considerably between bins (e.g., due to different threshold values).
To overcome this, our proof strategy is based on finding a sufficiently ``heavy" subset of bins that have roughly the same number of rejected balls in expectation. 

Towards that goal, for every bin $i$, we look at the value $S_i\coloneqq \mu+2\sqrt{\mu}-L_i$ and restrict attention to all bins satisfying that $S_i>0$. These bins are now divided into classes where, for $k\in \mathbb{Z}_{\geq 0}$, bin $i\in I_k\subseteq \{1,\ldots,n\}$ iff
$S_i\in [2^k,2^{k+1})$. Let $I^*$ be the class of all bins with $S_i\in (0,1)$. 

The selection of the class of bins for which we will show concentration is done in two steps.
First, we find at most $t\coloneqq \min\{\lceil\log n\rceil,\lceil\log (M/n)\rceil+1\}$ (plus 1) particular classes that together capture at least half of the expected value of rejected balls. Once we do that, we focus on the \emph{heaviest} class among these $t$ classes, hence loosing only a factor of $t$ in our bounds. Concretely, denoting by $k_{\max}$ the largest value of $k$ such that $I_{k_{\max}}\neq \emptyset$, the following holds.
\begin{claim}\label{cl:few_classes}
Let $k_{\min}\coloneqq \max\{k_{\max}-\lceil\log n\rceil+1,0\}$. Then the expected number of rejected balls by bins  $i \in [k_{\min}, k_{\max}]$ is at least $p_0\sqrt{Mn}/2$. In addition, $k_{\max}-k_{min}\leq t$.
\end{claim}
\begin{proof}
First, suppose that $k_{\max}\leq t$. Observe that the total contribution of 
all bins $i \in I^*$ is at most $n$, since $\sum_{i \in I^*}S_i \leq n$. By 
the prerequisite that $M\geq Cn$ for a sufficiently large constant $C$, we may 
assume that $C\geq 4/p_0^2$ and get that $n\leq \sqrt{Mn/C}\leq 
p_0\sqrt{Mn}/2$. As by \Cref{cor:exp} at least $p_0\sqrt{Mn}$ balls are 
rejected in expectation, the classes $1,\ldots, k_{\max}$ capture at least half 
of this expectation. 

Second, consider the case that $k_{\max}>t$. We claim that this entails that $t=\lceil\log n\rceil$, as $t=\lceil\log (m/n)\rceil+1$ would yield for all $i$ that
\begin{equation*}
\mu + 2\sqrt{\mu}-L_i\leq \mu + 2\sqrt{\mu}=\frac{M}{n}+2\sqrt{\frac{M}{n}}\leq \frac{2M}{n} \leq 2^t\,,
\end{equation*}
implying that $k_{\max}\leq t$. Therefore, indeed $t=\lceil\log n\rceil$ and hence
$k_{\min}= k_{\max}-t$. It follows that
\begin{equation*}
\sum_{i\in I^*}S_i + \sum_{k< k_{\min}}\sum_{i\in I_k}S_i
\leq  n \cdot \frac{2^{k_{\max}}}{n} \leq \sum_{i\in I_{k_{\max}}}
S_i\,.
\end{equation*}
Using the same expression for the expected number of rejected balls as in the proof of \Cref{cor:exp}, we get that
\begin{align*}
p_0\sum_{k=k_{\min}}^{k_{\max}}\sum_{i\in I_k} S_i \geq 
\frac{p_0}{2}\left(\sum_{i\in I^*}S_i + \sum_{k\in 
\mathbb{Z}_0}\sum_{i\in I_k}S_i\right) 
=\frac{p_0}{2}\sum_{i=1}^n\max\{\mu+2\sqrt{\mu}-L_i,0\}\geq 
\frac{p_0\sqrt{M/n}}{2}
\end{align*}
balls are rejected in expectation by bins in classes $k_{\min},k_{\min}+1,\ldots,k_{\max}$. As in the first case $k_{\max}-k_{\min}\leq t-0=t$ and in the second case $k_{\max}-k_{\min}=t$, this completes the proof.
\end{proof}

By the pigeonhole principle and \Cref{cl:few_classes}, there must be a class $k\in [k_{\min},k_{\max}]$ satisfying that
\begin{equation*}
p_0\sum_{i\in I_k} S_i \geq \frac{p_0\sqrt{Mn}}{2(t+1)}\,.
\end{equation*}
Denote by $z_i$, $i\in I_k$, the indicator variables that are $1$ iff
$X^{(i)}\geq \mu+2\sqrt{\mu}-L_i$. By \cite[Theorem 
13]{DubhashiR98} and \Cref{prop:monotone-na}
these variables are negatively associated. Setting $Z\coloneqq \sum_{i\in
I_k}z_i$, we have that $\E{Z}\geq p_0 |I_k|$, and by Chernoff's
bound~(Lemma~\ref{lem:chernoff}), it follows that
\begin{equation*}
\Pr\left[Z< \frac{p_0|I_k|}{2}\right]\leq e^{-\Omega(|I_k|)}\,.
\end{equation*}
If $|I_k|\geq (n/t)^{2/3}$, then we have that with probability $1-e^{-\Omega((n/t)^{2/3})}$, the number of
rejected balls is at least
\begin{equation*}
2^{k-1} p_0|I_k|\geq \frac{p_0}{4}\sum_{i\in I_k} S_i\in
\Omega\left(\frac{\sqrt{Mn}}{t}\right).
\end{equation*}
It remains to consider the case that $|I_k|< (n/t)^{2/3}$. 
Because up to factor $2$ all bins in $I_k$ have the same $S_i$ value, it holds for each $i\in I_k$ that
\begin{equation}\label{eq:li}
S_i=\mu+2\sqrt{\mu}-L_i\in \Omega\left(\frac{\sqrt{Mn}}{t \cdot |I_k|}\right).
\end{equation}
Let $\alpha\coloneqq \sqrt{\mu} \cdot n/(t \cdot |I_k|)> \sqrt{\mu}\cdot (n/t)^{1/3}>\sqrt{\mu}=\sqrt{M/n}$. By Inequality~\eqref{eq:li} and because $M\geq Cn$ for sufficiently large $C$,
\begin{equation*}
L_i\leq \mu+2\sqrt{\mu}-3\alpha \leq \mu-\alpha\,.
\end{equation*}
As $L_i\geq 0$, this bound also implies that $\delta\coloneqq \alpha/(2\mu) \in (0,1)$.
As $X^{(i)}$ is the sum of independent $0$-$1$ variables, we can thus apply
Chernoff's bound to $X^{(i)}$ to see that for sufficiently large $n$,
\begin{align*}
\Pr\left[X^{(i)}-L_i< \alpha/2\right]
&\leq \Pr\left[X^{(i)}\leq \mu-\alpha/2\right] \\
&\leq \Pr\left[X^{(i)}\leq \mu(1-\alpha/(2\mu))\right] \\
&\in e^{-\Omega(n^2/(t^2 \cdot |I_k|^2))} \in e^{-\Omega((n/t)^{2/3})}\,,
\end{align*}
where in the final step we use that $I_k<(n/t)^{2/3}$. By a union bound over all bins in $I_k$, we get that with probability $1-e^{-\Omega((n/t)^{2/3})}$, the number of rejected balls from this class is at least $\Omega(|I_k| \cdot \alpha)\subseteq \Omega(\sqrt{M n}/t)$.
\end{proof}
%
%

%% file: lower-simulation.tex
\subsection{Simulation for Higher Degree}\label{sec:degd}
In this subsection, we show that any algorithm with a higher degree (i.e., balls can contact more than one bin in a single round) can be simulated by an algorithm with degree 1 at the expense of more rounds. To this end, we simply increase the length of phases by factor $d$. We then proceed to show that a degree 1 algorithm with phase length $k>1$ can be improved on by reducing the phase length. We then can apply \Cref{thm:rejection} to the resulting degree 1 algorithm of phase length 1 to prove \Cref{thm:lower}.
\begin{lemma}\label{lem:generald}
Let $A$ be a uniform threshold algorithm of degree $d$ that runs in $r$ rounds. Then there is a uniform threshold 
algorithm $A'$ with degree 1 that achieves the same maximal load within $d \cdot r$ rounds.
\end{lemma}
\begin{proof}
$A'$ simulates $A$. It simply increases phase length by a factor of $d$ and lets the balls send their messages spread out over more rounds. This reduces the degree to $1$. At the end of each phase, the bins can compute the internal state they would have in $A$ and act accordingly. Thus, bin loads will be identical to those in $A$.
\end{proof}

\begin{lemma}\label{lem:reducek}
There is a uniform threshold algorithm of degree $1$ and phase length $1$ achieving the same guarantees on bin loads in the same number of rounds.
\end{lemma}
\begin{proof}
Assume that $A$ has phase length $k$. We simulate $A$ by algorithm $A'$ of phase length $1$. Balls and bins keep maintaining a state according to $A$, following these rules:
\begin{itemize}
  \item If a ball receives its first accept message in round $r$ of $A'$, it determines the phase $i=\lceil r/k\rceil$ of $A$ this round belongs to. Then it determines the phase $i'$ of $A$ in which it would inform bins about its decision. It will do so in $A'$ in round $i'k$ (i.e., the same round this would happen in $A$).
  \item For each $i\in \mathbb{N}$, at the beginning of round $(i-1)k+1$ each bin computes the threshold it would use in $A$ in phase $i$ based on the state for $A$ it maintains. This threshold is used in phases $(i-1)k+1,\ldots,ik$ of $A'$. The subset of balls it accepts in a given phase of $A'$ is chosen arbitrarily.
  \item To update the internal state a bin maintains for $A$ from phase $i$ to 
  phase $i+1$, at the end of round $ik$ it performs the following operation. 
  Let $P\subseteq \{1,\ldots,m\}$ be the set of ports it received requests on. 
  It determines the subset $Q\subseteq P$ of ports it would have responded to 
  with accept messages in $A$ when receiving the requests it got in rounds 
  $(i-1)k+1,\ldots,ik$. Let $Q'$ be the set of ports it sent accept messages to 
  in rounds $(i-1)k+1,\ldots,ik$ of $A'$. The bin now ``rearranges'' its port 
  numbering by permuting $P$ such that $Q'$ is mapped to $Q$. Finally, it 
  updates its state for $A$ in accordance with the modified port numbering and 
  the requests received during rounds $1,\ldots,ik$.
\end{itemize}
We claim that the third step maintains the invariant that the simulation is consistent with an execution of $A$ at the bin for the port numbering it computes. This holds true, because no bin ever sends two accept messages to the same ball, implying that the modification to the port numbering never conflicts with earlier such changes made. Thanks to this observation, a straightforward induction now establishes that $A'$ simulates an execution of $A$ for the port numberings the bins have determined by the end of the simulation. Accordingly, $A'$ achieves the same load distribution as $A$ with the modified port numbers.

Note that the choice of port numbers does not affect the guarantees on the load distribution $A$ makes, as we assumed an adversarial choice of bins' port numbers. Thus, the claim follows.
\end{proof}

We are now ready to complete the lower bound proof. 
\begin{proof}[Proof of Theorem~\ref{thm:lower}]
First, we show that the claim holds for degree $1$ algorithms with phase length $1$ by repeatedly applying \Cref{thm:rejection}. The induction hypothesis is that after round $i$, at least $M_i\coloneqq (m/n)^{3^{-i}}n^{1-3^{-i}}\in \omega(n)$ balls remain with probability $1-ie^{-\Omega(n^{1/2})}$. By the induction hypothesis, we have that
\begin{align*}
\min\{\log n,\log(M_i/n)\}\leq \log(M_i/n) \leq 
\log\left((m/n^2)^{3^{-i}}\right) \in O\left((m/n)^{3^{-(i+1)}/2}\right).
\end{align*}
As the total capacity of all bins is $n\cdot(m/n+O(1))=m+O(n)$ by assumption, the theorem\footnote{Note that we can apply Theorem~\ref{thm:rejection} due to the constraint that bins thresholds are independent from balls random choices regarding which bins to contact.} and the induction hypothesis imply that, with probability
\begin{align*}
\left(1-ie^{-\Omega(n^{1/2})}\right)\left(1-e^{-\Omega((n/\log 
n)^{2/3})}\right) \geq 1-(i+1)e^{-\Omega(n^{1/2})}\,,
\end{align*}
we have that
\begin{align*}
M_{i+1}\in \Omega\left(\frac{\sqrt{M_i n}}{\min\{\log 
n,\log(M_i/n)\}}\right)\subseteq  
\Omega\left(\left(\frac{(m/n)^{3^{-i}}n^{2-3^{-i}}}{(m/n)^{3^{-(i+1)}}}\right)^{1/2}\right)
 \subseteq  (m/n)^{3^{-(i+1)}}n^{1-3^{-(i+1)}}\,,
\end{align*}
as claimed.

Note that in the induction step we applied \Cref{thm:rejection}, which 
necessitates that $M_i\gg n$, which holds for sufficiently small $i\in 
\Omega(\log \log(m/n))$. To ensure that the probability bound is sufficiently 
strong for a w.h.p.\ result, we need, e.g., that $i\leq 2^{-n^{1/4}}\in 
2^{n^{\Omega(1)}}$. Both are ensured by the assumptions of the theorem. 
Finally, by applying \Cref{lem:generald} and \Cref{lem:reducek}, we can extend the result to degree 
$d$ algorithms for any $d=O(1)$ and arbitrary phase length $k$.
\end{proof}
%

%% file: assym.tex
\section{An Asymmetric Algorithm}\label{sec:apx-asym}
In this section, we prove \Cref{thm:assym} by providing an asymmetric algorithm 
that achieves a maximal load of $m/n + O(1)$, w.h.p., within a constant number of rounds. In this 
algorithm, each bin receives $O(m/n + \log n)$ messages in total. If $m>n \log n$, we apply a single round of the symmetric algorithm from \Cref{sec:parallel} first to reduce the number of remaining balls to $o(m)$, so that each bin receives $m/n+O(\log n)$ messages in the first round and $o(m)+O(\log n)$ messages in the subsequent application of the asymmetric algorithm.

Similarly to before, each active ball sends a single request in each round. The key idea of the algorithm is to operate on simulated ``superbins.'' Each superbin is controlled by a leader, where we make sure that the expected number $\mu$ of messages received by each superbin leader is roughly $m/n$ in each round (unless $m/n$ is very small). Denote by $\delta$ a value that is large enough so that the deviation from the expected number of messages a superbin receives is at most $\delta$ w.h.p. Then we can be sure that superbins receive $\mu-\delta$ messages w.h.p., and it allocates the respective balls to its bins round-robin.

As a result, the algorithm w.h.p.\ allocates \emph{exactly} the same number of balls to each bin, and it is straightforward to show that this process allocates all but $O(n)$ balls in a constant number of rounds. It then completes by invoking an asymmetric algorithm for allocating $n$ balls with constant load in constant time, where each bin simulates $O(1)$ virtual bins.

Concretely, the algorithm operates as follows.
\begin{enumerate}
  \item Set
  \begin{itemize}
    \item $m_1\coloneqq m$
    \item $r\coloneqq 1$.
  \end{itemize}
  \item Set
  \begin{itemize}
    \item $n_r\coloneqq m_r\min\{n/m,1/\log n\}$
    \item $\delta_r\coloneqq c\sqrt{m_r/n_r\cdot \log n}$ for a 
    sufficiently large constant $c$
    \item \begin{equation*}
    L_r \coloneqq \begin{cases}
    \lceil m_r/n_r - \delta_r\rceil & \mbox{if }\lceil m_r/n_r - \delta_r\rceil>2c^2\log n\\
    4c^2\log n & \mbox{else.}
    \end{cases}
    \end{equation*}
  \end{itemize}
  \item Each active ball chooses $i\in\{1,\ldots,n_r\}$ uniformly at random and contacts bin $i \cdot n/n_r$.
  \item Each bin selects up to $L_r$ requests and responds to them in a round-robin fashion with messages ``$j$'' for $j\in \{0,\ldots,n/n_r-1\}$.
  \item If a ball received response $j$ from bin $i$, it informs bin $i-j$ that it is allocated to this bin.
  \item If $L_r\neq \lceil m_r/n_r - \delta_r\rceil$, then terminate. Otherwise set\footnote{W.l.o.g., we assume that $n_{r+1}$ divides $n$; otherwise, one of the superbins is made at most factor $2$ larger, which does not affect the asymptotic bounds.}
  \begin{itemize}
    \item $m_{r+1}\coloneqq m_r-L_r n_r$
    \item $r\coloneqq r+1$
  \end{itemize}
  and go to Step~2.
\end{enumerate}
We establish the properties of the algorithm by a series of claims that are straightforward to show. First, we show that each superbin leader receives the ``right'' number of messages w.h.p.
\begin{claim}\label{cl:enough}
W.h.p., in round $r$ bins $i\cdot n/n_r$, $i\in \{1,\ldots,n_r\}$, receive between $m_r/n_r - \delta$ and $m_r/n_r+\delta$ messages (provided that $m_r$ is the number of unallocated balls at the beginning of the round).
\end{claim}
\begin{proof}
If $m\geq n \log n$, the expected number of messages per bin is $m_r/n_r=m/n\geq \log n$ and the claim is immediate from applying Chernoff's bound. Otherwise, this follows from a standard tail bound on the binomial distribution.
\end{proof}
\begin{claim}\label{cl:terminate}
The algorithm terminates in round $r$ iff $m_r/n_r\leq 2c^2\log n$.
\end{claim}
\begin{proof}
$\lceil m_r/n_r - \delta_r\rceil=m_r/n_r-c\sqrt{m_r/n_r \cdot \log n}\leq 2c^2\log n$ iff $m_r/n_r\leq 2c^2\log n$. Hence $L_r\neq \lceil m_r/n_r - \delta_r\rceil$ and the termination condition is satisfied iff this holds true.
\end{proof}
\begin{claim}\label{cl:3}
The algorithm terminates within $3$ rounds.
\end{claim}
\begin{proof}
Consider a round $r$ in which the algorithm does not terminate. By \Cref{cl:terminate}, thus $m_r/n_r>2c^2\log n$. Accordingly, $n_r=m_rn/m$ and $\delta_r=c\sqrt{m/n\cdot \log n}$. It follows that $m_{r+1}= m_r-L_rn_r\leq \delta_rn_r = m_r \sqrt{n/m\cdot \log n}$. If the algorithm does not terminate in the first two rounds, it follows that $m_3=m_1 \cdot n/m\cdot \log n = n\log n$. Therefore, $m_3/n_3=\log n <2c^2\log n$ and the algorithm terminates in round $3$. 
\end{proof}
\begin{claim}\label{cl:correct}
When the algorithm terminates, all balls are allocated w.h.p. The maximum bin load is $m/n+O(1)$ w.h.p.
\end{claim}
\begin{proof}
Consider a round $r$ in which the algorithm does not terminate. By \Cref{cl:enough}, superbin leaders receive at least $L_r=\lceil m_r/n-\delta \rceil$ messages w.h.p., implying that $n_rL_r$ balls are allocated in round $r$. By \Cref{cl:3}, the algorithm terminates within $3$ rounds. As $m_{r+1}=m_r-L_rn_r$ and $m_1=m$, a union bound thus shows that at the beginning of the final round $r\leq 3$, exactly $m_r$ unallocated balls remain w.h.p. Applying \Cref{cl:enough} to the final round, w.h.p.\ no bin receives more than $m_r/n_r+\delta$ messages. By \Cref{cl:terminate}, we have that $m_r/n_r\leq 2c^2\log n$ and thus $m_r/n_r+\delta \leq 4c^2\log n=L_r$. Hence, all balls are allocated w.h.p.

Concerning the bin load, observe that with the exception of the final round, loads cannot deviate by more than $1$ per round w.h.p., as each superbin receives exactly $L_r$ balls per round. However, in the final round we have that $L_r=4c^2\log n$. As $n_r\leq m_r/\log n$, each superbin consists of at least $\log n$ bins, so no bin receives more than $4c^2\in O(1)$ additional balls in this round.
\end{proof}
\begin{corollary}\label{cor:messages}
If $m\leq n\log n$, w.h.p.\ no bin receives more than $O(\log n)$ messsages. If $m>n\log n$, no bin receives more than $O(m/n)$ messages w.h.p.
\end{corollary}
\begin{proof}
By choice of $n_r$, we have that $m_r/n_r\leq \max\{m/n,\log n\}$ for each $r$. The corollary thus follows from \cref{cl:enough} if $m\leq n\log n$. If $m>n\log n$, we apply \cref{cl:enough} together with \cref{cl:3} and a union bound.
\end{proof}

\begin{proof}[Proof of \cref{thm:assym}]
\Cref{cl:3}, \Cref{cl:correct}, and \Cref{cor:messages} establish all the required claims except that bins receive $O(m/n+\log n)$ messages w.h.p.\ instead of $(1+o(1))m/n+O(\log n)$ w.h.p.\ in case $m>n\log n$. This is resolved by first executing a single round of the symmetric algorithm from \cref{sec:parallel}. The analysis shows that this allocates all but $o(m)$ balls such that most bin loads are the same; only $o(n)$ balls may be ``missing'' for a balanced allocation. Thus, using the asymmetric algorithm from this section to place the remaining $o(m)$ balls still guarantees a load of $m/n+O(1)$ w.h.p.\ and reduces the number of messages received by bins to $(1+o(1))m/n+O(\log n)$ w.h.p.
\end{proof}

%% file: conclusion.tex
\section{Conclusion}\label{sec:conclusion}

In this paper, we consider the somewhat neglected regime of $m \gg n$ in the setting of parallel balls into bins. We present a very simple algorithm that achieves the almost perfect load of $m/n+O(1)$ within $O(\log\log (m/n)+\log^* n)$ rounds with high probability. Our lower bound implies that the running time analysis of our algorithm is tight. This lower bound also generalizes to a broader class of threshold algorithms and implies for instance that allowing different threshold values does not increase the power of the algorithm. 

There are two intriguing open problems that we leave open. The first question concerns the upper bound: can we provide a faster symmetric algorithm for the problem? The second question concerns generalizing our lower bound argument to the entire class of threshold algorithms, i.e., removing the restriction that balls contact uniformly random bins. We conjecture that this is true whenever $n\gg \log \log (m/n)$, as then balls can glean little information out of having contacted $o(n)$ bins throughout the course of the algorithm; this would match the fact that, trivially, $n$ rounds are deterministically guaranteed by balls never contacting the same bin twice.